\documentclass[runningheads]{llncs}
\usepackage{bbm}
\usepackage{color}
\usepackage{enumerate}
\usepackage{mathrsfs}
\usepackage{xspace}
\usepackage{graphicx}
\usepackage{ifthen}
\usepackage{amsmath}
\usepackage{algorithm}
\usepackage{comment}
\usepackage{pstricks}
\usepackage{lmodern}
\usepackage[T1]{fontenc}

{\catcode`[=\active
\gdef\activatealgo{\let\>\qquad\let\=\leftarrow
        \catcode`[=\active\def[##1]{\ifmmode \lbrack##1]\else{\bf ##1}\fi}}
}

\newtheorem{observation}{Observation}
\newtheorem{prop}{Property}

\spnewtheorem*{main}{The Main Theorem}{\bfseries \upshape}{\itshape}
\spnewtheorem*{rrule}{Reduction Rule}{\bfseries \upshape}{\itshape}

\newcommand{\FPT}{\ensuremath{\text{FPT}}\xspace}

\newcommand{\Tlarge}{\ensuremath{\mathcal{T}_{\text{large}}}}
\newcommand{\Tsmall}{\ensuremath{\mathcal{T}_{\text{small}}}}

\newcommand{\Whard}[1]{\ensuremath{\text{W}[#1]\text{-hard}}}

\newcommand{\NPhard}{\ensuremath{\text{NP-hard}\xspace}}

\newcommand{\set}[1]{\ensuremath{\left \{ #1 \right \}}}

\newcommand{\tminor}{\ensuremath{\preceq_{\mathit{TM}}}}

\renewcommand{\qed}{\hfill \rule{2mm}{2mm}}

\newcommand{\tw}[1]{\mathop\mathbf{tw}(#1)}
\newcommand{\cw}[1]{\mathop\mathbf{cw}(#1)}
\newcommand{\rw}[1]{\mathop\mathbf{rw}(#1)}
\newcommand{\poly}{\mathop\mathit{poly}}

\newcommand{\equipi}[3]{\ensuremath{#1 \equiv_{\Pi,#3} #2}}

\newcommand*{\ie}{i.e.\@\xspace}

\newcommand*{\wrt}{w.r.t.\@\xspace}
\makeatletter
\newcommand*{\etc}{%
    \@ifnextchar{.}%
        {etc}%
        {etc.\@\xspace}%
}
\makeatother

\renewcommand\note[1]{}

\newcommand{\degav}{\ensuremath{d_{\text{av}}}}

\title{Linear Kernels on Graphs Excluding Topological Minors
\thanks{Research funded by DFG-Project RO 927/12-1 titled 
``Theoretical and Practical Aspects of Kernelization.''}
\thanks{A simpler proof of the results of this paper along with additional results 
appears in http://arxiv.org/abs/1207.0835}}
\author{Alexander Langer \and Felix Reidl \and Peter Rossmanith \and Somnath Sikdar}
 \institute{\normalsize
  Theoretical Computer Science, RWTH Aachen University, Germany \\
\email{\{langer|reidl|rossmani|sikdar\}@cs.rwth-aachen.de}}

\begin{document}
\maketitle

\begin{abstract}
We show that problems that have finite integer index and satisfy
a requirement we call \emph{treewidth-bounding} admit linear kernels
on the class of $H$-topological-minor free graphs, for an arbitrary
fixed graph~$H$. This builds on earlier results by Bodlaender et al.\ 
on graphs of bounded genus~\cite{BFLPST09} and by Fomin et al.\
on $H$-minor-free graphs~\cite{FLST10}. Our framework encompasses 
several problems, the prominent ones being \textsc{Chordal Vertex Deletion}, 
\textsc{Feedback Vertex Set} and \textsc{Edge Dominating Set}.
\end{abstract}

\section{Introduction}

Parameterized complexity deals with algorithms for decision problems whose
instances consist of a secondary measurement known as the \emph{parameter}.  
A major goal in parameterized complexity is to investigate whether a problem 
with parameter~$k$ admits an algorithm with running time $f(k) \cdot n^{O(1)}$.
Parameterized problems that admit such algorithms are called 
\emph{fixed-parameter tractable} and the class of all such problems is denoted \FPT.

A closely related concept is that of \emph{kernelization}. A kernelization 
algorithm for a parameterized problem takes as instance~$(x,k)$ of the problem
and, in time polynomial in $|x| + k$, outputs an equivalent instance~$(x',k')$ 
such that $|x'|, k' \le g(k)$, for some function~$g$. The function~$g$ 
is called the \emph{size of the kernel} and may be viewed as a measure
of the ``compressibility'' of a problem using polynomial-time preprocessing rules. 
By now it is a folklore result in the area that a decidable problem is 
fixed-parameter tractable iff it has a kernelization algorithm. What makes kernelization 
interesting is that many problems have a small kernel, meaning that the 
function~$g$ is polynomial or some times even linear.   

An important research direction is to investigate the parameterized 
complexity of problems that are \Whard{1}\footnote{The counterpart of \NPhard\ 
in parameterized complexity.} in general in special graph classes. It turns out (not surprisingly) 
that several \Whard{1} problems are not only in \FPT in special graph classes but admit 
linear kernels. A celebrated result is the linear kernel for \textsc{Dominating Set}
in planar graphs by Alber, Fellows, and Niedermeier~\cite{AFN04}. This paper prompted 
an explosion of research papers on linear kernels in planar graphs, including 
\textsc{Dominating Set}~\cite{AFN04,CFKX07}, \textsc{Feedback Vertex Set}~\cite{BP08},
\textsc{Cycle Packing}~\cite{BPT08}, \textsc{Induced Matching}~\cite{MS09,KPSX11},
\textsc{Full-Degree Spanning Tree}~\cite{GNW10} and \textsc{Connected Dominating Set}~\cite{LMS11b}.

Guo and Niedermeier showed that several problems that admit a 
``distance property'' admit linear kernels in planar graphs~\cite{GN07a}.
This result was subsumed by that of Bodlaender, Fomin, Lokshtanov, Penninkx, Saurabh and Thilikos in~\cite{BFLPST09} who
provided a meta-theorem for problems to have a linear kernel on graphs of bounded genus
(a strictly larger class than planar graphs). Later Fomin, Lokshtanov, Saurabh and Thilikos 
in~\cite{FLST10} extended these results for bidimensional problems to an even 
larger graph class, namely, $H$-minor-free and apex-minor-free graphs. The last two papers have provided 
deep insight into the circumstances under which problems admit linear (and polynomial)
kernels in sparse graphs. The property of \emph{finite integer index},
introduced by Bodlaender and van Antwerpen-de Fluiter~\cite{BvF01}, has emerged to be
of central importance to the aforementioned results: it guarantees the existence of
small gadgets that ``simulate'' large portions of the instance satisfying certain properties.
Finally note that a recent result by Fomin, Lokshtanov, Saurabh and Thilikos
now provides a linear kernel for \textsc{Dominating Set} and \textsc{Connected Dominating Set}
on $H$-minor-free graphs~\cite{FLST12}.

In this paper, we partially extend the results of Fomin et al.\ in~\cite{FLST10}
by giving a meta-result for linear kernels on $H$-topological-minor-free
graphs. More specifically, we show that any graph problem that has \emph{finite integer index} and   
is \emph{treewidth-bounding} has a linear kernel in $H$-topological-minor-free graphs. 
Informally, we call a problem treewidth-bounding if there exists a vertex set of small size 
whose deletion reduces the treewidth of the remaining graph to within a constant. 

Its worthwhile to note that Marx and Grohe have recently developed a decomposition theorem 
for $H$-topological-minor-free graphs along the same lines as the one 
for $H$-minor-free graphs~\cite{GM11}. As the latter proved to be extremely useful in designing
linear kernels for $H$-minor-free graphs, it would be very interesting to see
how one can apply this structure theorem to obtain kernels on graphs excluding
a fixed topological minor. Note, however, that for the results of this paper
we do not make use of this structure theorem.

The rest of the paper is organized as follows: Section~\ref{sec:Preliminaries} contains
the basic definitions and some important aspects of $H$-topological-minor-free graphs as
well as a key lemma used extensively in the proof of the main result. In Section~\ref{sec:MainResult} 
we present our main result, its implications in Section~\ref{sec:Implications}. Finally 
Section~\ref{sec:Conclusion} contains the conclusion and some open questions.

\section{Preliminaries}\label{sec:Preliminaries}

We use standard graph-theoretic notation (see~\cite{Die10} for any undefined terminology). 
Let~$e = xy$ be an edge in a graph~$G=(V,E)$. By~$G / e$, we denote the
graph obtained by \emph{contracting} the edge~$e$ into a new vertex~$v_e$, 
and making it adjacent to all the former neighbors of~$x$ and~$y$. 
A \emph{minor} of~$G$ is a graph obtained from a subgraph of~$G$
by contracting zero or more edges. A family~$\mathcal{F}$ of graphs 
is said to be \emph{minor-closed} if for all~$G \in \mathcal{F}$, every minor 
of~$G$ is contained in~$\mathcal{F}$. A graph~$G$ 
is said to be \emph{$H$-minor-free} if no minor of~$G$ is isomorphic to~$H$.
The class of $H$-minor-free graphs can be easily seen to be minor-closed. Note that
if~$G$ is $H$-minor-free then it is also $K_r$-minor-free, where~$r = |V(H)|$. 
Therefore no $H$-minor-free graph contains a clique with~$|V(H)|$ or more vertices. 
If a chordal graph~$G$ is $H$-minor-free, then every bag of the natural tree decomposition 
of~$G$ is a maximal clique of size at most~$r$.

Given a graph~$G=(V,E)$, a \emph{tree-decomposition of~$G$} is a pair~$(T, \mathcal{X})$,
where~$T$ is a tree and $\mathcal X = \{X_i \subseteq V(G) \mid i \in V(T)\}$ is a collection 
of vertex sets of~$G$ with one set for each node of the tree~$T$ such that the following hold:
\begin{enumerate}
\item $\bigcup_{i \in V(T)} X_i = V(G)$;
\item for every edge~$e = uv$ in~$G$, there exists~$i \in V(T)$ such that~$u,v \in X_i$;
\item for each vertex~$u \in V(G)$, the set of nodes~$\{i \in V(T) \mid u \in X_i\}$ 
induces a subtree. 
\end{enumerate}  
The vertices of the tree~$T$ are usually referred to as \emph{nodes} and the sets~$X_i$
are called \emph{bags}. The \emph{width} of a tree-decomposition is the size of the largest
bag minus one. The \emph{treewidth} of~$G$, denoted~$\tw{G}$, is the smallest width of a
tree-decomposition of~$G$.

Given a subtree $T' \subseteq T$ of a tree-decomposition 
$\mathscr{T} = (T, \mathcal{X})$ of a graph~$G$, the \emph{bags of $T'$} 
refer to the bags in $\mathcal{X}$ that correspond to the nodes in $T'$. 
We let $G[T']$ denote the graph induced by the vertices that occur in the bags
of $T'$. 

\subsection{Protrusions, $t$-Boundaried Graphs and Finite Integer Index}
In this subsection, we restate the definitions and results required for using 
the protrusion machinery used extensively in~\cite{BFLPST09,FLST10}. 

Given a graph~$G=(V,E)$ and a set~$W \subseteq V$, we define~$\partial_G(W)$
as the set of vertices in~$W$ that have a neighbor in~$V \setminus W$. For a
set~$W \subseteq V$ the neighborhood of~$W$ is $N_G(W) = \partial_G(V \setminus W)$.
Subscripts are omitted when it is clear which graph is being referred to.

\begin{definition}[$r$-protrusion~\cite{BFLPST09}]
Given a graph~$G=(V,E)$, we say that a set~$W \subseteq V$ is an $r$-protrusion
of~$G$ if $|\partial(W)| \leq r$ and $\tw{G[W]} \le r$. 
\end{definition}
If~$W$ is an $r$-protrusion, the vertex set~$W' = W \setminus \partial(W)$ is the 
\emph{restricted protrusion of~$W$}.

A \emph{$t$-boundaried graph} is a graph~$G=(V,E)$ with~$t$ distinguished vertices
labeled~$1$ through~$t$. The set of labeled vertices is denoted by~$\partial(G)$ 
and is called the \emph{boundary} or the \emph{terminals} of~$G$.  
For $t$-boundaried graphs $G_1$ and $G_2$, we let~$G_1 \oplus G_2$ denote the
graph obtained by taking the disjoint union of~$G_1$ and~$G_2$ and identifying
each vertex in~$\partial(G_1)$ with the vertex in~$\partial(G_2)$ with the same label.
This operation is called \emph{gluing}. 

\begin{definition}[Replacement]
	Let~$G=(V,E)$ be a graph with an $r$-protrusion $W$. Let~$W'$ be the
	restricted protrusion of~$W$ and let~$G_1$ be a $|\partial(W)|$-boundaried graph. 
	Then replacing~$G[W]$ by~$G_1$ corresponds to changing~$G$ into~$G[V \setminus W'] \oplus G_1$.
\end{definition}

We now restate the definition of one of the most important notions used in this paper.
\begin{definition}[Finite Integer Index~\cite{BvF01}] \label{def:finiteii}
	Let~$\Pi$ be a parameterized problem on a graph class~$\mathcal G$ and let~$G_1$
	and~$G_2$ be two $t$-boundaried graphs in~$\mathcal G$. We say that~$\equipi{G_1}{G_2}{t}$
	if there exists a constant~$c$ (that depends on $G_1$ and~$G_2$) such that for all $t$-boundaried
	graphs $G_3$ and for all~$k$:
	\begin{enumerate}
		\item $G_1 \oplus G_3 \in \mathcal G$ iff $G_2 \oplus G_3 \in \mathcal G$;
		\item $(G_1 \oplus G_3, k) \in \Pi$ iff $(G_2 \oplus G_3, k + c) \in \Pi$. 
	\end{enumerate}
	We say that the problem~$\Pi$ has finite integer index in the class~$\mathcal G$
	iff for every integer~$t$ the equivalence relation~$\equipi{}{}{t}$ has finite index.
\end{definition}

To test whether a parameterized problem has finite integer index on a graph class, one
can use the sufficiency test introduced in~\cite{FLST10} called \emph{strong monotonicity}.
We restate its definition for parameterized vertex deletion problems. An instance of a 
parameterized vertex deletion problem consists of a graph~$G$ and a parameter~$k$,
and the question is whether there exists a vertex set of size at most~$k$ whose deletion
results in a graph with some pre-specified property. Fix a vertex-deletion 
parameterized problem~$\Pi$ (analogous for an edge-deletion problem). Given $t$-boundaried graphs~$G, G'$ 
and~$X' \subseteq V(G')$, we let~$\zeta_G(G',X')$ denote the size
of the smallest vertex set~$X \subseteq V(G)$ such that~$X \cup X'$ is a 
solution to~$G \oplus G'$ for the problem~$\Pi$. If no such~$X$ exists, we 
define~$\zeta_G(G',X') = \infty$.
\begin{definition}
	A vertex deletion parameterized problem~$\Pi$ is strongly monotone if there exists a 
	function~$f \colon \mathbf{N} \rightarrow \mathbf{N}$ such that the following holds. 
	For every $t$-boundaried graph~$G = (V,E)$ there is a subset~$X \subseteq V$
	such that for every $t$-boundaried graph~$G' = (V',E')$ and~$X' \subseteq V'$
	such that $\zeta_G(G',X')$ is finite, $X \cup X'$ is a solution to~$(G \oplus G')$
	and we have $|X| \le \zeta_G(G',X') + f(t)$.
\end{definition}
Informally, a parameterized problem is strongly monotone if for every $t$-boundaried graph~$G$,
a local solution for~$G$ has nearly the same size as a global solution for~$G \oplus G'$
restricted to~$G$ for every $t$-boundaried graph~$G'$. 

It turns out that any graph-theoretic optimization problem where the objective
is to find a maximum or minimum sized vertex or edge set satisfying a 
(counting) MSO-predicate has finite integer index if it is
strongly monotone. 
\begin{proposition}{\rm (\cite{BFLPST09}, see Lemma~12)}\label{prop:StronglyMonotone}
	Every strongly monotone \textsc{$p$-min-CMSO} and \textsc{$p$-max-CMSO} problem has
	finite integer index.
\end{proposition} 

We adapt the notion of \emph{quasi-compact} problems introduced in \cite{BFLPST09}
for graphs of bounded genus to that of \emph{treewidth bounding} problems by removing
the radial distance, which is not applicable to the more general class of graphs
excluding a fixed topological minor.

\begin{definition}[Treewidth Bounding]
	A parameterized graph problem $\Pi$ is called treewidth bounding if for 
        every $(G,k) \in \Pi$ it holds that there exists a set $X \subseteq V(G)$ such that
	\begin{enumerate}
		\item $|X| \le c k$, and 
		\item $\tw{G-X} \le t$,
	\end{enumerate}
    where~$c,t$ are constants that depends only on~$\Pi$. 
	We call a problem \emph{treewidth bounding
	on a graph class $\mathcal G$} if the above property holds under the restriction that
	$G \in \mathcal G$.
\end{definition}
For problems whose solution is a vertex subset, the set~$X$ will be the solution set 
of $(G,k)$. For simplicity we will call the set~$X$ \emph{the} solution in the following.

\subsection{Properties of $H$-topological-minor-free graphs}
In this section we list some properties of $H$-topological-minor-free graphs 
that we use in the proofs to follow. We use~$r$ to denote~$|V(H)|$.

The first property states that graphs that exclude a fixed graph~$H$ as 
a topological minor are sparse in some sense.
\begin{prop}[{\rm \cite{BT98},\cite{KS96a}}]\label{prop:HFreeDegree}
	The average degree $\degav$ in an $H$-topological-minor-free graph is bounded by
	$\degav < \beta r^2$ for some $\beta \leq 10$.
\end{prop}
As a corollary, a graph with average degree larger than~$\beta r^2$ 
contains~$H$ as a topological minor. 

It is clear that if a graph excludes~$H$ as a topological minor, then it does
not have~$K_r$ as a topological minor. What is also true is that the \emph{total} number
of cliques (not necessarily maximal) is linear in the number of vertices.   
\begin{prop}[{\rm \cite{FOT10}}]\label{prop:HFreeCliques}
	There is a constant $\tau < 4.51$ such that, for $r > 2$, every $n$-vertex graph with no
	$K_r$ topological minor has at most $2^{\tau r \log r} n$ cliques.
\end{prop}

\begin{definition}\label{def:XDegree}
	Let $G$ be a graph and $X,Y \subseteq V(G)$ two disjoint vertex sets of
	$G$. Then we define \emph{the degree of $Y$ with respect to $X$} as 
	$$
		D_X(Y) = \left| \set{u \in X \mid v \in Y : uv \in E(G)}  \right|. 
	$$
\end{definition}
We will sometimes be sloppy with our notation and, for a subgraph~$G'$ of~$G$, 
write~$D_X(G')$ instead of~$D_X(V(G'))$.

One technique frequently used in the proofs that follow is embodied in the proof of the 
following lemma.

\begin{lemma}\label{lemma:LargeDegreeComponents}
	Let~$G$ be an $H$-topological-minor-free graph, let $X \subseteq V(G)$,
	and $C_1,\dots,C_l$ be pairwise vertex-disjoint 
	connected subgraphs of $G-X$ such that $D_X(C_i) \geq r \geq 2$.
	Then $l \leq \frac{1}{2} \beta r^2 |X|$.
\end{lemma}
\begin{proof}
	We construct a topological minor $G' \tminor G$ such that each edge in $G'$ 
	corresponds to a subgraph $C_i$. The construction works as follows. 
        Delete all edges in the graph~$G[X]$. For each connected subgraph~$C_i$, 
        choose distinct vertices~$x,y \in X$ such that~$xy$ is not an edge and 
        both~$x$ and~$y$ are adjacent to~$u$ and~$v$ in~$C_i$, respectively. We explicitly
		allow the case $u = v$.

        Next choose a path~$P_{uv}$ from~$u$ to~$v$ in~$C_i$ and delete all vertices of~$C_i$
        save those from~$P_{uv}$. Finally contract the path~$P_{uv}$ to an edge between~$x$
        and~$y$. This sequence of operations clearly produces a topological minor since
        the only edges that were contracted had at least one endpoint with degree at most two.
        
        Since topological minor containment is a transitive, the graph~$G'$ obtained 
        by ``contracting each connected subgraph~$C_i$ into an edge'' is also $H$-topological-minor-free.
        Observe that since we assumed that~$D_X(C_i) \geq r \geq 2$, for each component~$C_i$,
        there exists distinct vertices~$x,y \in X$ that are adjacent to~$C_i$ and which 
        do not yet have an edge between them. If this were not the case, the neighbors of $C_i$ 
	in $X$ form a clique of size at least~$r$. We would then have an $r$-clique 
	in a topological-minor of~$G$, contradicting the fact that it 
        is $H$-topological-minor-free.
        It now follows that the number~$t$ of subgraphs~$C_i$ is bounded by the number
	of edges in $G'$. By Property~\ref{prop:HFreeDegree}, the number of edges 
	is linear in the size of~$X$ and we obtain the following bound:
	$$
		l \leq |E(G')| \leq \frac{1}{2} \beta r^2 |V(G')| = \frac{1}{2} \beta r^2 |X|.
	$$
\qed
\end{proof}

\section{Main result}\label{sec:MainResult}

In this section we prove our main result.
\begin{main}
Fix a graph~$H$. Let~$\Pi$ be a parameterized graph-theoretic problem 
on the class of $H$-topological-minor-free graphs that has finite integer index
and is treewidth-bounding. Then $\Pi$ admits a linear kernel.
\end{main}

Let~$(G,k)$ be a yes-instance of~$\Pi$, where~$G$ is $H$-topological-minor-free. Since 
we assumed that~$\Pi$ is treewidth-bounding there exists~$X \subseteq V(G)$ such 
that~$\tw{G - X} \le t$, where~$t$ is a constant that depends only on~$\Pi$. 
Since the problem~$\Pi$ is assumed to have finite integer index, denote by 
$\varpi(i)$ the size of the largest representative of the equivalence relation~$\equipi{}{}{i}$,
where the representatives are chosen such that they are smallest possible. We use only 
one reduction rule which is stated below.
\begin{rrule}[Protrusion Reduction Rule]
Let~$W \subseteq V(G)$ be a protrusion with $|\partial(W)| \leq 2t+r$ 
such that the restricted protrusion~$|W'|$ has size strictly more than $\varpi(2t+r)$.
Let~$G'$ be the representative of~$G[W]$ in the equivalence relation $\equipi{}{}{|\partial(W)|}$. 
Replace $G$ by $G[V \setminus W'] \oplus G'$ and the parameter~$k$ by~$k - c$. 
\end{rrule}
Here~$c$ is the constant in the definition of finite integer index (see Definition~\ref{def:finiteii}) 
that depends on~$G'$ and~$G[W]$.

From now on whenever we talk about an instance~$(G,k)$ of the problem~$\Pi$, 
we assume that it is reduced \wrt our only reduction rule.
In particular, $G$ does not contain a $(2t+r)$-protrusion 
of size strictly more than~$\varpi(2t+r)$. 
\begin{definition}
Let~$(G,k)$ be a yes-instance of~$\Pi$ and let~$X \subseteq V(G)$ be such 
that~$\tw{G - X} \le t$, where~$t$ is a constant. 
Let $\mathscr{C}_S$ and $\mathscr{C}_L$ denote, respectively, the set of all connected 
components~$C$ of~$G-X$ such that $D_X(C) < r$ and $D_X(C) \geq r$. Call the components
of~$\mathscr{C}_S$ ``small'' and those of~$\mathscr{C}_L$ ``large.'' 
\end{definition} 

\subsection{Bounding the size of small components}

We first bound the total number of vertices in all components in~$\mathscr{C}_S$.
\begin{lemma}\label{lemma:SmallDegreeComponents}
	The total number of vertices in all components in $\mathscr{C}_S$ is bounded from
	above by 
	$\varpi(r) (2^{\tau r \log r} + \beta r^2) k$. 
\end{lemma}
\begin{proof}
	First note that for each $C \in \mathscr{C}_S$, the set $Y = N(C) \cap X$ 
	of its neighbors in~$X$ is a separator of~$C$ in~$G$ of size at most~$r$. 
	Therefore the total number of vertices in all components separated by~$Y$ in~$G$ 
	is at most~$\varpi(r)$, where~$\varpi(r)$ is the size of the largest representative of 
	the equivalence relation~$\equipi{}{}{r}$. Thus it is sufficient to show that 
	the number of subsets of~$X$ that are separators of components in $\mathscr{C}_S$ is bounded.

	Using the technique outlined in the proof of Lemma~\ref{lemma:LargeDegreeComponents}, 
        we contract the components $\mathscr{C}_S$ greedily into edges in $X$. Repeat the following
	operations for as long as possible. Pick a component $C \in \mathscr{C}_S$ arbitrarily 
	and choose two distinct vertices $u,v \in N(C) \cap X$ such that~$uv$ is not an edge;
        create a new edge~$uv$, and delete~$C$ from the graph. If there are components~$C$ 
	which cannot be contracted into edges in this fashion, then it follows that the 
	separators~$N(C) \cap X$ of these components are cliques in $X$. As the subgraph~$G_X$ 
	induced by~$X$ after these operations is a topological minor of~$G$, it 
	must be  that~$G_X$ is $H$-topological minor free. Hence by Property~\ref{prop:HFreeCliques}, 
	$G_X$ has at most~$2^{\tau r \log r}k$ cliques. The number of vertices in 
	all components of $\mathscr{C}_S$ separated by such a clique is, as noted before, 
	$\varpi(r)$. Moreover each component that is contracted to an edge also has 
	at most~$\varpi(r)$ vertices and, by Property~\ref{prop:HFreeDegree}, $G_X$ 
	has at most $\beta r^2 k$ edges. Hence the total number of vertices in all components of~$\mathscr{C}_S$ 
	is bounded from above by $\varpi(r)(\beta r^2 + 2^{\tau r \log r})k$. \qed
\end{proof}

\subsection{Bounding the size of the large components}

Proving that the total number of vertices in all components of~$\mathscr{C}_L$ 
is linear in~$k$ is more involved. As a first step, we use 
Lemma~\ref{lemma:LargeDegreeComponents} to show that the \emph{number} 
of components in~$\mathscr{C}_L$ is linear in~$k$. To bound the total number 
of vertices in~$\mathscr{C}_L$ as a linear function of~$k$, we propose a technique 
of decomposing components~$C \in \mathscr{C}_L$ into connected subgraphs each of 
bounded size but with a ``large'' number of neighbors in the set~$X$. 
The following structure plays a crucial role in bounding the size of $G-X$.

\begin{definition}[Scrubs and Twigs]
	Let~$G=(V,E)$ be a graph and let~$X \subseteq V$ be such that~$\tw{G - X} \le t$,
	for some constant~$t$. A scrub~$\mathscr{S}$ in~$G$ is a pair $(R, \mathcal{W})$, 
	where $R \subseteq V \setminus X$ and $\mathcal{W}$ is a maximal family of 
	vertex-disjoint sets $W_1, \ldots, W_l \subseteq V \setminus (X \cup R)$ 
        each of which induces a connected subgraph and such that the following conditions hold:
	\begin{enumerate}
		\item for~$1 \leq i \leq l$, $D_X(W_i) < r$;
		\item $R \cup W_1 \cup \ldots W_l$ induces a connected subgraph in $G - X$.
	\end{enumerate}	
	We call~$R$ the root and~$\mathcal{W}$ the twigs of the scrub~$\mathscr{S}$.
	The size of~$\mathscr{S}$ is defined as $|\mathscr S| = |R \cup W_1 \cup \cdots \cup W_l|$.
\end{definition}

In what follows we let $\mathscr{T}_C = (T_C,\mathcal{X}_C)$ denote a tree-decomposition of 
$C \in \mathscr{C}_L$ that is \emph{rooted} at some arbitrary bag of degree at least two 
in the decomposition. We define $\mathscr{F}$ to be the ``forest-decomposition'' 
obtained by taking the disjoint union of all tree-decompositions, that is, 
$$\mathscr{F}:= \left (\bigcup_{C \in \mathscr{C}_L} T_C, \bigcup_{C \in \mathscr{C}_L}
\mathcal{X}_C \right) = (F, \mathcal{X}).$$ 

We will employ a marking algorithm that marks bags of $\mathscr F$ to demonstrate 
that the total number of vertices in all the components $\mathcal C_L$ is 
indeed bounded in a reduced instance. We stress however that this algorithm
is not efficient, and neither does it have to be, since it is only used 
to show that the kernel size is small. In what follows, we let 
$\mathscr{M} \subseteq \mathcal X$ denote the set of \emph{bags} that have already 
been marked by the algorithm and $V(\mathscr{M})$ to be the set of all 
\emph{vertices of the graph} which occur in at least one marked bag. Call a subtree 
of some tree-decomposition in~$\mathscr{F}$ \emph{marked} if it contains 
at least one marked bag and \emph{unmarked} otherwise. Note that an unmarked
subtree~$T'$ can contain marked vertices of~$G$ in its bags, as these vertices 
could occur in some other marked bag. 

The marking algorithm works as follows.
\begin{enumerate}
	\item Set~$\mathscr{M} := \emptyset$.
	\item Mark bags $B$ of the forest decomposition $\mathscr{F}$ 
	      which induce a scrub $\mathscr S = (B, \mathcal W)$ 
              in the graph $G-(X \cup V(\mathscr{M}))$ that satisfies the
              following conditions:
		  \begin{itemize}
		  	\item $|\mathscr S| > \varpi(t+r)$;
			\item $D_X(\mathscr S) \geq r$.
		  \end{itemize}
	Set $\mathscr{M} = \mathscr{M} \cup \{B\}$.
	
	\item Mark join bags $J$ that are parents of unmarked subtrees that induce at 
              least one connected component in $G - (X \cup J)$ with at least $r$
              neighbors in~$X$. Add each such bag to~$\mathscr{M}$.
	
	\item Iteratively mark the least common ancestor (join) bag of two bags 
	      that have already been marked and add it to~$\mathscr{M}$.
\end{enumerate}

We first point out some features of the marking algorithm. 
\begin{lemma}\label{lemma:markingFeatures}
If~$T'$ is a subtree of a tree~$T \in \mathscr{F}$ such that each bag in~$T'$ is unmarked,
then at most two of the neighboring bags of~$T'$ in~$T - V(T')$ are marked.
\end{lemma}
\begin{proof}
Suppose that~$T'$ is unmarked but has at least three marked bags as neighbors in~$T - V(T')$. 
Let~$\mathscr{P}$ be the shortest path from the root of~$T$ to~$T'$. If the root of~$T$ 
happens to be in~$T'$, then~$\mathscr{P}$ consists of only the root bag. Now 
there are at least two marked bags~$B_1,B_2$ in $T - V(T')$ that are neighbors 
of~$T'$ that are \emph{not} on~$\mathscr{P}$. Clearly one of the bags of~$T'$ 
must be the least common ancestor of~$B_1$ and~$B_2$ and the algorithm, in Step~4, would
then have marked this bag. This contradicts the hypothesis that~$T'$ has no marked
bags. 
\qed
\end{proof}

\begin{lemma}\label{lemma:numberMarkedBags}
The total number of bags marked by the algorithm is at most~$2 \beta r^2 k$ and therefore
$|V(\mathscr M)| \leq 2 \beta r^2 kt$.
\end{lemma}
\begin{proof}
Since~$\tw{G - X} \leq t$, each bag of the tree-decomposition has size at most~$t$ (we assume
an optimal tree-decomposition). To prove the lemma, it is sufficient to bound the 
number of bags marked by the algorithm in Steps~2, 3, and~4. 

The scrubs that are marked in Step~2 are vertex-disjoint and since each 
scrub ``sees'' at least~$r$ vertices in the set~$X$, by 
Lemma~\ref{lemma:LargeDegreeComponents}, the number of such scrubs is at 
most~$\beta r^2 k/2$. In Step~3, the connected components that are 
considered are vertex-disjoint and hence the bound of Lemma~\ref{lemma:LargeDegreeComponents} 
applies again. Finally in Step~4, the number of marked bags doubles in the worst case. 
This proves the bound on the size of~$V(\mathscr{M})$.    
\qed
\end{proof}

Lemma~\ref{lemma:numberMarkedBags} showed that the total number of vertices in 
marked bags is linearly bounded in~$k$. We now go on to show that the total 
number of vertices in \emph{unmarked bags} is also linear in~$k$.
To achieve this goal, we first consider the total size of the scrubs seen 
by the algorithm in Step~2. Suppose that in this step, the algorithm 
considers the scrubs $\mathscr S_1, \ldots, \mathscr S_l$ in that order while marking bags.

\begin{lemma}\label{lemma:ScrubsCount}
The total number of vertices in the scrubs $\mathscr S_1, \ldots, \mathscr S_l$ is bounded from
above by $\beta^2 r^4 t \varpi(t+r)k$.
\end{lemma}
\begin{proof}
	Let~$\mathscr{S}_i = (R_i, \mathcal{W}_i)$ and consider a twig $W \in W_i$ 
	for some $1 \leq i \leq l$. By the definition of a scrub, $D_X(W) < r$, 
	and hence $G[W]$ is separated from the rest of the graph by the set 
	$R_i \cup (N(W) \cap X)$ which is of size at most $t+r$. It follows that 
	in a reduced instance $|W| \leq \varpi(t+r)$.
	In fact, the total number of vertices in \emph{all} twigs of a scrub~$\mathscr{S}_i$ 
	that are connected to the same set of vertices in $X$ is bounded $\varpi(t+r)$---these 
	twigs share a common separator.

	Also note that the scrubs $\mathscr{S}_i$ are vertex-disjoint and therefore by 
	Lemma~\ref{lemma:LargeDegreeComponents} it follows that~$l \le \beta r^2 k/2$. 
	Therefore in order to bound the total number of vertices in all the scrubs, 
        it is sufficient to bound the total number of twigs. 
	Let $S = \bigcup_i R_i$. Construct a bipartite graph $G'$ from $G$ 
	with bipartition $S \uplus X$ and edge set $E_{SX}$ as follows:
	\begin{enumerate}
	\item Delete all vertices of $G$ that are \emph{not} in either~$X$ nor 
	in any scrub $\mathscr{S}_i$.
	\item Delete all edges inside the root~$R_i$ of scrub $\mathscr{S}_i$, for $1 \le i \le l$. 
	\item Delete all twigs $W \in \mathcal W_i$ that have no neighbors in~$X$.
	\item For all twigs in $W_i$ that are connected to the same set in $X$, remove all but one.
	\item For each twig $W \in \mathcal W_i$, choose arbitrary vertices 
	      $u \in R_i \cap N(W)$ and $v \in X \cap N(W)$. Remove $W$ and add 
              the edge $uv$ to $E_{SX}$.
	\end{enumerate}
	
	Now $|S| \leq l \cdot t \leq \beta r^2 k t/2$. For each scrub $\mathscr{S}_i$, 
	the number of vertices in the twigs removed in Step~3 is at most $\varpi(t)$ 
	and hence the total number of vertices removed in this step over all scrubs 
	is bounded from above by $\beta r^2 k \varpi(t)/2$. For each scrub $\mathscr{S}_i$
	and each subset~$X' \subseteq X$, the number of vertices in the twigs removed in 
	Step~4 is bounded from above by $\varpi(t+r)$.

	The bipartite graph $G' = (S \uplus X, E_{SX})$ is a topological minor of $G$, 
	and since $G$ is $H$-topological-minor-free, so is $G'$. By Property~\ref{prop:HFreeDegree},
	the number of edges in~$G'$ is at most 
	$$
		|E_{SX}| \leq \beta r^2 (|S|+|X|) \leq \beta r^2 (\frac{1}{2} \beta r^2 kt + k). 
	$$
	The total number of vertices removed in Step~4 is therefore 
	$\beta r^2 (\beta r^2 t/2 + 1) \cdot \varpi(t+r)k$.
        It follows that the total number of vertices in the scrubs 
	$\mathscr S_1, \dots, \mathscr S_l$ is bounded from above by 
	\begin{eqnarray*}
	& & \frac{1}{2}\beta r^2 k \varpi(t) + \beta r^2 \left ( \frac{1}{2}\beta r^2 kt 
          + k \right) \cdot \varpi(t+r) \\
	& \leq & \beta r^2 k \left ( \frac{\varpi(t)}{2} + \left (
	\frac{1}{2}\beta r^2 t + 1 \right ) \cdot \varpi(t+r) \right) \\	
	& \leq & \beta^2 r^4 t \varpi(t+r) k.
	\end{eqnarray*}
	\qed
\end{proof}

At this point, we have accounted for all vertices that occur in a marked bag
or a scrub $\mathscr{S}_1, \ldots, \mathscr{S}_l$ seen by the algorithm in Step~2. 
We now consider the forest-decomposition $\mathscr{F}'$ obtained from $\mathscr{F}$ 
by removing all vertices that occur in marked bags. This corresponds to a 
forest-decomposition of the graph $G -(X \cup V(\mathscr{M}))$. 
Note that we may not remove all the scrub vertices in this process. In
order account for the fact that all vertices in the scrubs $\mathscr{S}_1, \ldots \mathscr{S}_l$
have been counted, we simplify the forest-decomposition $\mathscr{F}'$ even further.
Delete a tree $T \in \mathscr{F}'$ if all its bags \emph{only} contain scrub vertices
from $\mathscr{S}_1, \ldots, \mathscr{S}_l$. 

The trees in the forest-decomposition can be partitioned into two classes: those
that have at most~$r - 1$ neighbors in $X$ and those that have at least~$r$ neighbors. 
This motivates us to define $\Tsmall$ and $\Tlarge$.
Define $\Tsmall$ to be the set of trees $T \in \mathscr{F}'$ such 
that $D_X(G[T]) \leq r-1$; $\Tlarge$ is the set of trees 
$T \in \mathscr{F}'$ such that $D_X(G[T]) \geq r$. By Lemma~\ref{lemma:markingFeatures},
at most two neighboring bags of a tree $T \in \Tsmall \cup \Tlarge$ 
are marked.

\begin{lemma}\label{lemma:PseudoScrub}
	Let $B$ be a marked bag in the forest-decomposition $\mathscr{F}'$ and let $B_1$, \ldots, $B_p$ 
	be its neighbors such that the subtrees $T_1, \ldots, T_p$ rooted at these 
	bags satisfy $D_X(G[T_i]) < r$, for $1 \le i \le p$. 
	Then the total number of vertices in the bags of $T_1, \ldots, T_p$ that do not appear	
	in any of the scrubs $\mathscr{S}_1, \ldots, \mathscr{S}_l$ is bounded by $\varpi(t+r)$.	
\end{lemma}
\begin{proof}
	Let $V_i$ be the vertices of $G$ that are contained in the bags of $T_i$. 
        Note that $B$ cannot be the root of a scrub $\mathscr S_j$ found by the algorithm 
	in Step~2, otherwise $V_1,\dots,V_p$ would be in some scrub $\mathscr S_1,\dots,\mathscr S_j$. 
	If $B \subseteq V(\mathscr S_1) \cup \dots \cup V(\mathscr S_l)$, then all $V_i$
	must be contained in some scrub $\mathscr S_j$. Therefore 
        $B' = B \setminus \bigcup_i^l V(\mathscr S_i)$
	cannot be empty if $p > 0$. But then $(B,\set{V_1,\dots,V_p})$ is a scrub in 
	$G-(X \cup V(\mathscr M))$ not chosen by the algorithm in Step~2. 
	Since the algorithm chooses scrubs of size at least $\varpi(t+r)$, this 
	implies that $|B \cup V_1 \cup \dots \cup V_p \setminus \bigcup_i^l V(\mathscr S_i)| 
        \leq \varpi(t+r)$. \qed
\end{proof}

We next show that the total number of vertices in the trees in $\Tsmall$ is linear
in~$k$.
\begin{lemma}\label{lemma:RestSmall}
	The total number of vertices in the bags in $T \in \Tsmall$ 
	that do not appear in the scrubs $\mathscr{S}_1, \ldots, \mathscr{S}_l$ 
	is at most $4 \beta r^2  \varpi(2t+r) k$. 
\end{lemma}
\begin{proof}
	By Lemma~\ref{lemma:markingFeatures}, at most two neighboring bags of a tree 
	$T \in \Tsmall$ are marked. Therefore for each $T \in \Tsmall$, 
	the number of vertices in the bags of $T$ is at most $\varpi(2t+r)$, 
	as the subgraph $G[T]$ has a separator of size at most $2t+r$. 
        Moreover the number of trees in $\Tsmall$ that have exactly \emph{two} marked bags 
	as neighbors is bounded by the number of marked bags (we can simply associate 
	each such tree with one marked bag in the forest $F$). We therefore have to 
	bound the number of trees that have exactly one marked bag as neighbor. 
	By Lemma~\ref{lemma:PseudoScrub}, the total number of vertices in trees of $\Tsmall$ 
	that are adjacent to exactly one marked bag~$B$ is at most $\varpi(t+r)$. 
        Since the number of marked bags is at most $2\beta r^2 k$, 
	the total number of vertices in bags of $T \in \Tsmall$ is at most 
	$2 \beta r^2 (\varpi(2t+r) + \varpi(t+r))k$ which is at most $4 \beta r^2 \varpi(2t+r)k$.
	\qed
\end{proof}

All that now remains is to show that the vertices in trees of $\Tlarge$ 
that have not been accounted for thus far is linear in $k$.

\begin{observation}\label{obs:numComponents}
Let~$\mathscr{T}$ be a tree-decomposition of a connected 
graph $G$. Let~$B$ be some bag of $\mathscr{T}$ and let $B_1, \ldots, B_p$
be some of its neighbors. Let~$T_1, \ldots, T_p$ be the subtrees of~$\mathscr{T}$ 
rooted at $B_1, \ldots, B_p$, respectively, and let~$V_1, \ldots, V_p$ be the 
vertices of $G$ that occur in the bags of these subtrees. Then the graph 
$G[B \cup V_1 \cup \cdots \cup V_p]$ has at most~$|B|$ connected components.
\end{observation}

Similar to Lemma~\ref{lemma:PseudoScrub}, we have the following:
\begin{lemma}\label{lemma:SmallScrubs}
	Let~$B$ be an unmarked bag in the forest-decomposition $\mathscr{F}$ 
	and let $T_1, \ldots, T_p$ be the unmarked 
	subtrees that are rooted at the neighbors of $B$ in $\mathscr F$ such that
	$D_X(T_i) < r$ for $1 \leq i \leq p$; let $V_1, \ldots, V_p$ be the vertices of~$G$
	that appear in the bags of $T_1, \ldots, T_p$, respectively. Then
	the number of unmarked vertices in $W = B \cup V_1 \cup \cdots \cup V_p$ 
	is bounded from above by $t \varpi(t+r)$.
\end{lemma}
\begin{proof}
	If $D_X(W) < r$, then as in 
	Lemma~\ref{lemma:SmallDegreeComponents}, we have $|W| \leq \varpi(t+r)$. 
	Therefore assume that $D_X(W) \geq r$ and $|W| \geq \varpi(t+r)$. Since $B$ was not marked 
	in Step~2, for all $C \subseteq W \setminus V(\mathscr M)$ that induce connected 
        components in $G-X$, it holds that $|C| \leq \varpi(t+r)$. By Observation~\ref{obs:numComponents}, 
	$G[W]$ can have at most $t$ connected components and hence the claimed bound follows.
	\qed
\end{proof}

\begin{definition}[Central Path]
	Let $T \in \Tlarge$ be adjacent to two marked bags $B_i, B_j$. 
	The central path of $T$ is the unique path from bag $B_i$ to 
        bag $B_j$ in $T$. If $T$ is adjacent to only one marked bag $B_i$,
	the central path is defined as a path $P$ from $B_i$ to a leaf of $T$ 
	such that $D_X(G[P])$ is maximized.
\end{definition}

\begin{lemma}\label{lemma:widthPathDecom}
	If $T \in \Tlarge$ then $G[T]$ has a path decomposition of width at 
	most $t(\varpi(t+r)+1)$.
\end{lemma}
\begin{proof}
	Let $P$ be the central path of $T$. Construct a path-decomposition of $G[T]$
	as follows. Take all bags in the path $P$ and, for each join bag $B$ on this path, 
	add in the vertices of all bags connected to $B$ that are not part of $P$.
	As each such join bag is unmarked, by Lemma~\ref{lemma:SmallScrubs}, the 
	size of such a bag increases by at most $t\varpi(t+r)$. As the size of 
	each bag of $P$ is bounded by $t$, the above bound follows. \qed
\end{proof}

We next show that if $T \in \Tlarge$ and if the subgraph $G[T]$ induced by the 
vertices in the bags of $T$ is large, then we can decompose it into \emph{connected subgraphs} 
$G'$ of constant size such that $D_X(G') \geq r$. Lemma~\ref{lemma:LargeDegreeComponents} 
assures us that there can be at most $O(k)$ such connected subgraphs. Together this would
imply a linear bound on the total number of vertices in $\Tlarge$. 

To state this ``decomposition lemma,'' we introduce additional notation and terminology. 
Given a path decomposition $\mathscr{P} = (P, \mathcal X)$ of a graph~$G$ 
and two bags $X, Y \in \mathcal{X}$, let $G(X,Y)$ denote the graph induced by the vertices in 
the bags that appear between $X$ and $Y$ in $\mathscr{P}$ excluding the vertices in $X$ and $Y$. 
That is, if $\mathcal{B}$ denotes the set of bags in the path $P$ starting with $X$ and ending 
with $Y$, then
\[G(X,Y) =  G \left [ \bigcup_{B \in \mathcal{B}} B \setminus (X \cup Y) \right ].\]
The first and last bag of $\mathscr{P}$ are called its \emph{end-bags}.
Given a bag $Z \in \mathcal{X}$, we say that $G(X,Y)$ is connected to $Z$ if it either
includes a vertex from $Z$ or is adjacent to a vertex in $Z$.
Let $T$ be a tree-decomposition of a graph $G$ and let $A,B$ be bags occurring
in $T$. For a subtree $T' \subseteq T$, we say that $G[T']$ has an $AB$-path 
(or, a path from $A$ to $B$) if there exists a $uv$-path 
in $G[T']$ where $u \in A$ and $v \in B$. This trivially holds if 
$A \cap B \neq \emptyset$ and $G[T']$ contains a vertex of $A \cap B$.
In the following lemma, we write $f(r,t)$ for the expression 
$(3t \varpi(t+r) + t) \cdot \varpi(2t+r) + t(\varpi(t+r) + 1)$ as a shorthand.

\begin{lemma}[The Cutting Up Lemma]\label{lemma:CutUp}
	Let $T \in \Tlarge$ and let $\mathscr{P} = (P, \mathcal{X})$ be 
        a path-decomposition of $G[T]$ with $A$ and $B$ as its end-bags. If $G[T]$ has 
	an $AB$-path then either it has at most $f(r,t)$ vertices or there exists 
	a bag $Z \in \mathcal{X}$ such that the following hold: 
	\begin{enumerate}
		\item $G(A,Z)$ has at most $f(r,t)$ vertices 
		      and contains a connected component $C$ with $D_X(C) \geq r$;  
		\item either $D_X(G(Z,B)) \geq r$ or $|G(Z,B)| \leq \varpi(2t+r)$.
	\end{enumerate}
\end{lemma}
\begin{proof}
	Let~$p$ be the width of the path-decomposition $\mathscr{P}$. By 
        Lemma~\ref{lemma:widthPathDecom}, this is at most $t (\varpi(t+r) +1)$.
	We first show that for any bag $Z \neq A$ in the decomposition $\mathscr{P}$, 
	the graph $G(A,Z)$ contains at most $p + 2t \varpi(t+r)$ connected components. 
        Note that each connected component of $G(A,Z)$ is connected to either 
        $A$ or $Z$. If this were not the case, then the tree-decomposition in 
        $\mathscr F$ of which $T$ is a subtree would contain more than one connected component of 
	$G-X$. This is a contradiction since we assume that each tree in the forest $\mathscr{F}$
	represents a connected component.

	The number of connected components of $G(A,Z)$ connected to \emph{both} $A$ and $Z$ 
	is bounded by the width~$p$ of the decomposition. To see this, simply observe that
	the graph $\tilde{G}(A,Z)$ obtained from $G(A,Z)$ by adding edges such that both~$A$ and~$Z$ 
	induce cliques also has pathwidth at most~$p$. If the number of connected 
	components in $G(A,Z)$ connected to both $A$ and $Z$ were at least $p+1$ then at least 
	$p+2$ cops would be required to catch a robber in $\tilde{G}(A,Z)$, 
	contradicting the fact that it has pathwidth~$p$. The number of components 
        connected exactly to one of $A$ or $Z$ is, by Lemma~\ref{lemma:SmallScrubs}, 
        at most $p+2t \varpi(t+r)$. 

	Imagine walking along the bags of the decomposition $\mathscr{P}$ starting at~$A$
	and suppose~$Z$ is the first bag such that $|G(A,Z)| \geq (p+2t \varpi(t+r)) \cdot \varpi(2t+r)$.
	Then $G(A,Z)$ contains a connected component~$C$ with at least $\varpi(2t+r)$ vertices. 
        Since our instance is reduced, it must be that $D_X(C) \geq r$. Let~$Z'$ be the bag 
	immediately before~$Z$. Then 
	\[
	  |G(A,Z)| \leq |G(A,Z')| + |Z| \leq (p + 2t \varpi(t+r)) \cdot \varpi(2t+r) + p,
	\]
	and since $p \leq t (\varpi(t+r) +1)$, an easy calculation shows that 
	$|G(A,Z)| \leq f(r,t)$ proving claim~$(1)$ of the lemma. Claim~$(2)$ 
        is easier to show. For if $D_X(G(A,Z)) < r$, then $G(Z,B)$ 
	has a separator of size at most $2t+r$ and, since the graph is reduced, 
        has at most $\varpi(2t+r)$ vertices.
	\qed
\end{proof}

Finally, we can bound the number of vertices occurring in the bags of trees in~$\Tlarge$.

\begin{lemma}\label{lemma:sizeLargeTrees}
	The total number of vertices in bags of all trees $T \in \Tlarge$ 
        is $O(k)$. 
\end{lemma}
\begin{proof}
	If~$T \in \Tlarge$ has at least $f(r,t)$ vertices then using Lemma~\ref{lemma:CutUp}, 
        we can decompose $G[T]$ iteratively into connected components~$G'$ of size at most 
	$f(r,t)$ with $D_X(G') \geq r$. By Lemma~\ref{lemma:LargeDegreeComponents},
	the total number of connected components is at most $\beta r^2 k/2$. 
	Finally by Lemma~\ref{lemma:CutUp} the number of vertices in all the bags
	of trees in $\Tlarge$ is at most $\beta r^2 k/2 \cdot (f(r,t) + \varpi(2t+r))$,
	which is at most $6 \beta r^2 t \varpi(2t+r)^2 k$.
\qed
\end{proof}

We are now ready to prove the Main Theorem.
\begin{main}
Fix a graph~$H$. Let~$\Pi$ be a parameterized graph-theoretic problem 
on the class of $H$-topological-minor-free graphs that has finite integer index
and is treewidth-bounding. Then $\Pi$ admits a linear kernel.
\end{main}

\begin{proof}
Let $(G,k)$ be a yes-instance of $\Pi$ that has been reduced \wrt the Protrusion Reduction Rule.  
Using Lemmas~\ref{lemma:SmallDegreeComponents}, \ref{lemma:numberMarkedBags},
\ref{lemma:ScrubsCount}, \ref{lemma:RestSmall}, and~\ref{lemma:sizeLargeTrees} 
we see that $|V(G)| = O(k)$. 
\qed
\end{proof}

This result immediately extends to graphs of bounded degree, as graphs of maximum 
degree~$d$ cannot contain~$K_{d+1}$ as a topological minor.
\begin{corollary}
Let~$\Pi$ be a parameterized graph-theoretic problem that has finite integer index
and is treewidth-bounding, both on the class of graphs of maximum degree~$d$. 
Then $\Pi$ admits a linear kernel.
\end{corollary}

\section{Implications of the Main Theorem}\label{sec:Implications}

Some concrete problems that fall under this definition are the following.
\begin{corollary}\label{cor:concreteProblems}
Fix a graph~$H$. The following problems are treewidth-bounding and have  
finite integer index on the class of $H$-topological-minor-free graphs 
and hence admit a linear kernel on this graph class. 
\textsc{Vertex Cover};\footnote{Listed for completeness; these problems have 
a linear kernel on general graphs.} 
\textsc{Cluster Vertex Deletion};\footnotemark[\value{footnote}]	
\textsc{Feedback Vertex Set}; 
\textsc{Chordal Vertex Deletion}; 
\textsc{Interval} and \textsc{Proper Interval Vertex Deletion};
\textsc{Cograph Vertex Deletion}; 
\textsc{Edge Dominating Set}.	
\end{corollary}

This also implies that, by a simple brute force on the kernelized
instance, \textsc{Chordal Vertex Deletion} and \textsc{Interval Vertex Deletion}
are solvable in $O(c^k \poly(n))$ time for some constant $c$. On general graphs
only a $f(k) \poly(n)$ algorithm is known~\cite{Mar10}. 

As an example of a concrete class of problems that satisfy the Main Theorem,
consider a hereditary property $\mathcal{P}$ whose forbidden set
contains all holes. Any graph that satisfies $\mathcal{P}$ must necessarily be 
chordal. The \textsc{$\mathcal{P}$-Vertex Deletion} problem is, given a graph~$G$ 
and an integer~$k$, to decide whether there exists at most~$k$ vertices whose 
deletion results in a graph satisfying $\mathcal{P}$. It is easy to show that 
\textsc{$\mathcal{P}$-Vertex Deletion} is both treewidth-bounding and 
strongly monotone (and hence has finite integer index) on $H$-topological-minor-free 
graphs and therefore admits a linear kernel on such a graph class.  

A natural extension of the problems in Corollary~\ref{cor:concreteProblems} 
is to ask for a \emph{connected} solution. In many cases, however,
the connected version of a problem is not strongly monotone and 
probably does not have finite integer index. For the following problems, however, 
strong monotonicity can be shown easily as any solution contains vertices 
at a constant distance from the boundary.
\begin{corollary}
\textsc{Connected Vertex Cover}, \textsc{Connected Cograph Vertex Deletion} 
and \textsc{Connected Cluster Vertex Deletion} have linear kernels in 
graphs excluding a fixed topological minor.
\end{corollary}

An interesting property of $H$-topological-minor-free graphs is that the usual 
width measure are essentially the same.
\begin{prop}[{\rm \cite{FOT10}}]\label{prop:HFreeWidth}
	There is a constant\footnote{This is the same constant $\tau$ as used in Proposition~\ref{prop:HFreeCliques}} 
	$\tau$ such that for every $r > 2$, if $G$ is a graph excluding $K_r$ as a topological minor, then
	$$\begin{aligned}
		\rw G \leq \cw G & <  2 \cdot 2^{\tau r \log r} \rw G \\
		\rw G \leq \tw G + 1 & < \frac{3}{4}(r^2 + 4r - 5)	2^{\tau r \log r} \rw G \\  
	\end{aligned}$$
\end{prop}

\noindent This entails the following Corollary of the main result.

\begin{corollary}
	The problem of deleting $k$ vertices such that the remaining graph has bounded clique-, tree- or
	branchwidth has a linear kernel in graphs excluding a fixed topological minor.
\end{corollary}

Finally, we can relate our result to bidimensionality in
some natural cases.  Consider a vertex-deletion problem
\textsc{$\mathcal{P}$-Vertex Deletion} for some arbitrary graph
property $\mathcal{P}$. Then our result entails the following.

\begin{corollary}
If \textsc{$\mathcal P$-Vertex Deletion} has finite integer index and is
bidimensional, then it has a linear kernel on graphs excluding a fixed topological minor.
\end{corollary}

\begin{proof}
Let $(G,k)$ be a yes-instance with solution set $X \subseteq V(G)$.
Then $G-X \in \mathcal P$, which entails that, for some constant
$c$ depending only on $\mathcal P$, $G-X$ does not contain a
$c \times c$-grid as a minor. Otherwise the solution of $G-X$
would be nonempty: if we could contract $G-X$ into a
grid that itself is not in $\mathcal P$, \ie we need to delete at
least one vertex from it to obtain a graph that has the property
$\mathcal P$, this would contradict the assumption that the problem
is bidimensional.

\noindent Therefore \textsc{$\mathcal P$-Vertex Deletion} is
treewidth-bounding and the above follows.  \qed
\end{proof}

\section{Conclusion and Open Questions}\label{sec:Conclusion}

We have shown that one can obtain linear kernels for a range of problems
on graphs excluding a fixed topological minor. This partially extends the results by
Bodlaender et al.\ on graphs of bounded genus~\cite{BFLPST09} and by Fomin et al.\ 
on graphs excluding a fixed minor~\cite{FLST10}.

Two main questions arise: (1)~can similar results be obtained for an even
larger class of (sparse) graphs and (2)~what other problems have linear kernels on 
$H$-topological-minor free graphs. In particular, does \textsc{Dominating Set}
have a linear kernel on graphs excluding a fixed topological minor?
It would also be interesting to investigate how the structure theorem by
Grohe and Marx can be used in this context~\cite{GM11}.

\bibliographystyle{abbrv}
\bibliography{cross,conf}
\end{document}